\documentclass[10pt,journal]{IEEEtranTCOM}
\usepackage{graphicx} 
\usepackage[english]{babel}
\usepackage{amsthm}
\usepackage{amsmath}
\usepackage{amsfonts}
\usepackage{mathtools}

\newtheorem{theorem}{Theorem}[section]

\newcommand{\be}{\begin{equation}}
\newcommand{\ee}{\end{equation}}
\newcommand{\la}{\langle}
\newcommand{\ra}{\rangle}

\title{No-cloning theorem for 2WQC and postselection}
\author{\IEEEauthorblockN{Mah Noor}\\
\IEEEauthorblockA{Higher Education Department, Punjab, Pakistan, \emph{anwar.mahnoor@yahoo.com}}\\
\IEEEauthorblockN{Jarek Duda}\\
\IEEEauthorblockA{Faculty of Mathematics and Computer Science, Jagiellonian University, Krakow, Poland, \emph{dudajar@gmail.com}}}
\date{22 July 2024}

\begin{document}
\maketitle
\begin{abstract}
Two-way quantum computers (2WQC) are proposed extension of standard 1WQC: adding conjugated state preparation operation $\langle 0|$ similar to postselection $|0\ra \langle 0|$, by performing a process which from perspective of CPT symmetry is the original state preparation process, for example by reversing EM impulses used for state preparation. As there were concerns that this extension might violate no-cloning theorem for example for attacks on quantum cryptographic protocols like BB84, here we extend the original proof to show this theorem still holds for 2WQC and postselection.    
\end{abstract}

\section{Introduction}
No-cloning theorem~\cite{no-cloning} is seen as a crucial property of quantum physics, for example being at heart of security of quantum cryptographic protocols like BB84~\cite{BB84}. 

Its standard version focuses on one-way quantum computers 1WQC~\cite{1WQC}: using state preparation $|0\ra$, unitary evolution and measurement. However, recently there was proposed 2WQC extension~\cite{2WQC}: adding conjugated state preparation $\la 0 |$ to the allowed set of operations. Suggested realization is by exploiting CPT symmetry, which switches them: $|0\ra \xleftrightarrow{\text{CPT}} \la 0|$, and is required by local Lorentz invariant quantum field theories~\cite{CPT}. 

Therefore, performing a process which from CPT symmetry perspective becomes the original state preparation process, we should get its conjugated version. For example for silicon quantum dots~\cite{dots} state preparation is just EM impulses, hence using reversed EM impulses at the end, from CPT symmetry perspective they become the original state preparation process.

Such $\la 0|$ operation mathematically acts as postselection $|0\ra \la 0|$, however, it replaces selection requiring multiple runs, with physical constraints exactly as in state preparation, in theory allowing the same calculations in a single run.

In theory 2WQC are more powerful, for example allowing to solve postBQP~\cite{postBQP} problems including NP in a single run, or for faster and more stable Grover algorithm~\cite{Grover}.

As 2WQC extends the set of available operations, limitations proven for 1WQC generally might not necessarily hold - they should be revisited, extending the proofs, or maybe finding counterarguments. Here we show that no-cloning theorem still holds, forbidding also postselection-based attacks on protocols like BB84. In contrast, for example for Closed Timelike Curves (CTC) there are claims of possible state cloning~\cite{CTC}. 

\section{No-cloning theorem and extension to 2WQC}
For intuitions we will start with simplified proofs for single qubits, then generalize and formalize.
\subsection{Single qubit simplified proofs}\label{simple}
\subsubsection{1WQC case} Assume we have $U$ two qubit  operator able to clone any single qubit state: $U|v\ra|0\ra=|v\ra|v\ra$, and let us show it leads to contradiction. 

\noindent For example it has to clone $|0\ra, |1\ra$ basis states:
\be U|0\ra|0\ra=|0\ra|0\ra \qquad\qquad U|1\ra|0\ra=|1\ra|1\ra \ee
Decomposing $|\phi\ra=a|0\ra +b|1\ra$ in this basis we get:
\be U|\phi\ra | 0\ra= |\phi\ra |\phi\ra = a^2 |0\ra|0\ra +ab |0\ra|1\ra +ba |1\ra|0\ra + b^2|1\ra|1\ra \label{eI}\ee
However, directly using linearity of $U$, we get instead:
\be U |\phi\ra |0\ra =U(a|0\ra +b|1\ra)|0\ra=
a|0\ra|0\ra+b|1\ra |1\ra \label{eII}\ee
Comparing (\ref{eI}) and (\ref{eII}) equations above, we get $ab=0$, hence at least one of them is zero, what is in contradiction with assumption that $U$ can clone any state.\\

\subsubsection{2WQC case}Here we add third qubit initialized to $|0\ra$, which is finally projected/postselected to $\la 0|$. This (partial) projection can reduce amplitude generally depending on cloned $|\phi\ra=a|0\ra +b|1\ra$, which additionally would need to contain some $0<|c_\phi|\leq 1$ amplitude. Finally, for $\la 0|$ projecting to 3rd qubit, the assumed cloning operator $U$ would require below:
\be \la 0|U|\phi\ra |0\ra |0\ra = c_\phi |\phi\ra |\phi\ra\quad\textrm{for some}\quad 0<|c_\phi|\leq 1\ee
Analogously to 1WQC case, we get the below equations:
\be \la 0| U|0\ra|0\ra|0\ra=c_0 |0\ra|0\ra \qquad \la 0|U|1\ra|0\ra|0\ra=c_1|1\ra|1\ra \ee
and the (\ref{eI}) and (\ref{eII}) equations become analogous:
\be \la 0|U|\phi\ra| 0\ra|0\ra=\label{sim}\ee
$$=c_\phi|\phi\ra |\phi\ra = c_\phi(a^2 |0\ra|0\ra +ab |0\ra|1\ra +ba |1\ra|0\ra + b^2|1\ra|1\ra)$$
$$=\la 0|U|(a|0\ra +b|1\ra)| 0\ra|0\ra=c_0 a  |0\ra|0\ra + c_1 b  |1\ra|1\ra $$
Their equality still requires $ab=0$, being in contradiction with assumption that $U$ is universal cloning operator.

\subsection{General version with proof}
Below is standard (1WQC) formulation of no-cloning theorem (usually $|e\rangle=|00\ldots 0\rangle$), on Hilbert space $H=H_A=H_B$, which for cloning needs to be equal for both states. We will further extend it to 2WQC and prove as generalized.
\begin{theorem}No-cloning (1WQC): There is no unitary operator U on $H\otimes H$ and normalized $|e\rangle_B$ in $H$, such that for all normalized states $|\phi\rangle_A\in H$:
\begin{equation} U\left(|\phi\rangle_A |e\rangle_B\right)=e^{i\alpha_\phi}\,|\phi\rangle_A |\phi\rangle_B
\end{equation}
for some real number $\alpha_\phi$ depending on $\phi$.
\end{theorem}

For 2WQC we have available additional third set of qubits - generally from Hilbert space $H_C$ (empty to reduce to 1WQC), usually initialized to $|f\rangle_C=|00\ldots 0\rangle\in H_C$, and finally projected to usually the same $|g\rangle_C=|00\ldots 0\rangle\in H_C$. 

Another difference is that instead of the $|e^{i\alpha_\phi}|=1$ coefficient, postselection/conjugated state preparation can reduce amplitude: for 2WQC we use $c_\phi\in \mathbb{C}$ satisfying $0<|c_\phi|\leq 1$. It degenerates to 1WQC case for empty $H_C$ and $c_\phi=e^{i\alpha_\phi}$. 

This amplitude reduction is caused by its part going to a state orthogonal to $|g\ra_C$, below denoted as $|\Phi_\phi\ra$, generally depending on the state to clone $\phi$. For normalization $|c_\phi|^2+\la \Phi_\phi|\Phi_\phi\ra=1$. This $|\Phi_\phi\ra$ is zeroed by postselection/conjugated state preparation, requiring to finally normalize the state.

\begin{theorem} No-cloning for 2WQC: There is no unitary operator U on $H\otimes H\otimes H_C$ and normalized $|e\rangle_B \in H$, 
$|f\rangle_C, |g\rangle_C\in H_C$, such that for all normalized states $|\phi\rangle_A\in H$:
\begin{equation} \la g|_C\, U\left(|\phi\rangle_A |e\rangle_B\right|f\rangle_C)=c_\phi |\phi\rangle_A |\phi\rangle_B 
\end{equation}
or equivalently:
\begin{equation} U\left(|\phi\rangle_A |e\rangle_B\right|f\rangle_C)=\alpha |\phi\rangle_A |\phi\rangle_B |f\rangle_C + |\Phi_\phi\rangle \
\end{equation}
for some complex number $0<|c_\phi|\leq 1$ depending on $\phi$, also  $|\Phi_\phi\rangle$ in $H\otimes H\otimes |g\rangle^{\perp}$ with third component orthogonal to $|g\rangle_C$.
\end{theorem}
\begin{proof} Let us assume existence of such $U$ and lead to contradiction. As in simplified proof above, let us expand $\phi$ in orthonormal basis $B_H$ of $H=H_A=H_B$, e.g. 0/1 qubits:
\be |\phi\rangle =\sum_{v\in B_H} \phi_v |v\ra \ee
The $U$ operator should allow to clone any $v$ from $B_H$ basis:
\be \la g|_C\, U\left(|v\rangle_A |e\rangle_B\right|f\rangle_C)=c_v |v\rangle_A |v\rangle_B 
\ee
Now as in (\ref{sim}), let us clone $\phi\in H$ in two ways: 
\be \la g|_C\, U\left(|\phi\rangle_A |e\rangle_B\right|f\rangle_C)=\ee
$$=c_\phi |\phi\ra_A |\phi\ra_B = c_\phi \sum_{u,v\in B_H} \phi_u \phi_v |u\ra_A |v\ra_B $$
$$=\la g|_C\, U\left(\sum_{v\in B_H} \phi_v |v\ra_A \right)|e\rangle_B|f\rangle_C=\sum_{v\in B_H} \phi_v c_v |v\ra_A |v\ra_B$$
Comparing both, for $u\neq v$ we get $\phi_v \phi_u = 0$, allowing for at most one $\phi_v$ coordinate being nonzero, what is in contradiction with claim that $U$ can clone any $\phi\in H$.
\end{proof}

\section{Conclusions and further work}
This article concludes discussion on no-cloning theorem for 2WQC and postselection  - in contrast to CTC~\cite{CTC}, remaining valid here, for example forbidding this type of attacks on quantum cryptographic protocols like BB84, also attacks trying to clone some percentage of qubits using postselection.

As further work it should be generalized like for 1WQC e.g. for mixed states~\cite{mixed}, also analyzed for imperfect cloning~\cite{imperfect}. Also there should be investigated other limitations of 1WQC if they remain valid for 2WQC, what generally not needs to be true as it extends the set of available operations.

\begin{figure}[t!]
    \centering
        \includegraphics{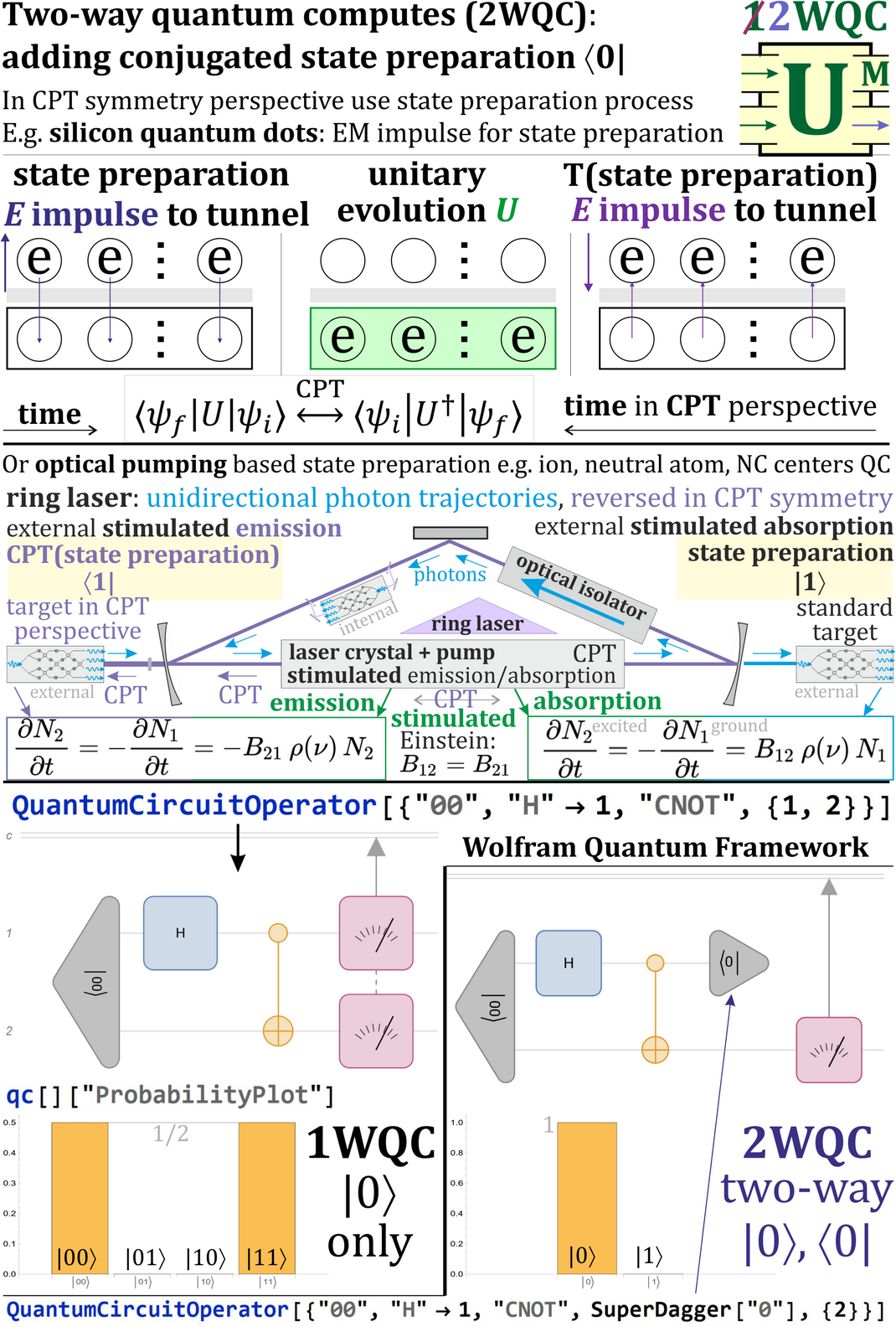}
        \caption{2WQC~\cite{2WQC} proposes to extend 1WQC set of allowed operation by conjugated state preparation $|0\rangle$ (like postselection but by physical constraints) - performing state preparation process in CPT perspective, e.g. by reversing used EM impulses (top), or using ring laser with optical isolator (center). Bottom: simple example of implementation in Wolfram Quantum Framework. }
        \label{2WQC}
\end{figure}
\section*{Acknowledgments} 
This article was written during QInter 2024 program of QWorld, which helped with organization of this collaboration.

\bibliographystyle{IEEEtran}
\bibliography{cites}
\end{document}